%% file: main.tex
\documentclass{article}

\PassOptionsToPackage{numbers,sort&compress}{natbib}
\usepackage[dblblindworkshop,final]{neurips_2026}
\workshoptitle{SaTQuML: Secure and Trustworthy Quantum Machine Learning}

\usepackage[utf8]{inputenc}
\usepackage[T1]{fontenc}

\usepackage{amsmath,amssymb,mathtools,amsthm}
\usepackage{booktabs}
\usepackage{microtype}
\usepackage{xcolor}
\usepackage{graphicx}
\usepackage{url}
\usepackage{hyperref}
\usepackage{pgfplots}
\usepgfplotslibrary{groupplots}
\pgfplotsset{compat=1.18}

\input{generated/results_macros.tex}
\input{generated/endpoint_rows.tex}
\input{generated/contrast_rows.tex}
\input{generated/curve_rows.tex}
\input{generated/positive_results_macros.tex}

\input{generated/positive_endpoint_rows.tex}
\input{generated/positive_contrast_rows.tex}
\input{generated/positive_curve_rows.tex}
\input{generated/sensitivity_control_macros.tex}

\input{generated/matched_control_rows.tex}
\input{generated/classical_control_rows.tex}
\input{generated/gain_transplant_macros.tex}

\input{generated/gain_transplant_rows.tex}
\input{generated/head_refit_macros.tex}

\definecolor{fullblue}{HTML}{275DAD}
\definecolor{ebgred}{HTML}{C33C54}
\definecolor{fgqagreen}{HTML}{2E8B57}
\hypersetup{
  colorlinks=true,linkcolor=fullblue,citecolor=fullblue,urlcolor=fullblue,
  pdftitle={When Normalization Selects the Sign: Auditing Robustness Ablations in Quantum Attention},
  pdfauthor={Owen Friedewald; Srikar Alla; Ali Shiri Sichani; Chi-Ren Shyu}
}
\newcommand{\method}{\textsc{SeqFormer}}
\newcommand{\ebg}{\textsc{EBG}}
\newcommand{\fgqa}{\textsc{FGQA}}
\newcommand{\nach}{\textsc{NACH}}
\newcommand{\Rdiag}{R_{\mathrm{diag}}}
\newtheorem{proposition}{Proposition}

\title{When Normalization Selects the Sign:\\
Auditing Robustness Ablations in Quantum Attention}

\author{%
  \href{https://orcid.org/0009-0003-6847-1706}{Owen Friedewald}\quad
  \href{https://orcid.org/0009-0008-1870-5304}{Srikar Alla}\quad
  \href{https://orcid.org/0000-0002-7239-1341}{Ali Shiri Sichani}\quad
  \href{https://orcid.org/0000-0001-9197-9522}{Chi-Ren Shyu}\\
  Department of Electrical Engineering and Computer Science\\
  University of Missouri, Columbia, MO, USA\\
  \texttt{\{omfvq4,saffw,asp9f,shyuc\}@missouri.edu}
}

\begin{document}
\maketitle

\begin{abstract}
Removing an input-scaling module changes both a classifier and the
perturbations reaching its encoder.  A robustness difference can therefore
reflect the comparison rule as well as the module.  We demonstrate this
problem in a four-qubit quantum-attention detector on generated power-grid
trajectories.  A learned scaling module appears beneficial at a fixed physical
attack budget, but matching an upper bound on perturbations at the encoder
reverses the ordering.  Neither comparison alone establishes a robustness
benefit caused by the module.  The initial test also perturbs clean examples
into attacked examples while retaining their original labels; tests restricted
to already attacked examples do not establish a benefit.  Replacing a trained
model's input scales disrupts detection.  Retraining its linear classification
layer restores the detection rate, but changes individual predictions, leaving
the comparison descriptive rather than causal.  Two further design checks
explain why the input quantum Fisher information regularizer cannot train
this model's query parameters, and why removing confidence bounds does not
establish a larger certified radius.  The evidence is limited to ten seeds,
exact simulation, synthetic data, and a restricted set of attacks; classical
baselines achieve better clean prediction.  The practical lesson is to specify
which perturbation budget is fixed, check that attacks preserve labels and
interventions preserve predictions, and distinguish exploratory controls from
confirmatory evidence.
\end{abstract}

\section{Introduction}

Adversarial robustness is unusually easy to overstate in quantum machine
learning (QML).  A small classifier can appear robust because an attacker was
defined after dimensionality reduction, because a stochastic gradient was not
audited, or because point probabilities were treated as certificates.  These
problems are familiar in classical robustness evaluation
\citep{athalye2018obfuscated,croce2020reliable,madry2018towards}, but QML adds
finite-shot estimation, costly parameter-shift gradients, and ambiguity about
which distance in circuit-input space corresponds to a realizable physical
change.  Meanwhile, quantum classifiers are themselves vulnerable to small
adversarial perturbations \citep{lu2020quantum,liu2020vulnerability}; evidence
of quantum-enhanced robustness remains sensitive to architecture, threat
model, and comparison design \citep{west2023towards}.
The neighboring causal-intervention literature likewise shows that attribution
depends on whether an intervention preserves the behavior being compared and
on which response metric is used \citep{geiger2022inducing,zhang2024towards}.

We study these issues through \method{}, a compact quantum self-attention
classifier for temporal false-data-injection attack (FDIA) detection.
Attacks act on raw measurements through the MATPOWER IEEE 14-bus direct-current (DC) model
\citep{zimmerman2011matpower}, with two controllable state directions and a
feature reducer fitted only on training data.  This keeps the physical intervention explicit throughout
the comparison.

The experiment isolates three study-specific modules, rather than modifying
a published robustness defense.  The encoding bandwidth
governor (\ebg{}) learns bounded feature scales and directly controls the
upper bound on the encoder's sensitivity to input changes.  Fisher-geometric quantum attention
(\fgqa{}) penalizes quantum Fisher information (QFI) inside the mapped attack
subspace while weakly rewarding information outside it.  The noise-aware
certification head (\nach{}) applies simultaneous Clopper--Pearson confidence
bounds before computing a radius diagnostic; in exact statevector simulation we use
this only to check the confidence-bound calculation, not a hardware certificate.  We compare the full model with
paired $-\ebg$, $-\fgqa$, $-\nach$, and all-modules-off conditions over ten
seeds.  A minus sign denotes removal of that module; we call this a knockout.

The central result is a normalization-dependent ablation ordering, not an
identified EBG benefit.  The initial test includes clean and attacked examples;
adding attacks can invalidate the clean examples' labels.  The later test on
already attacked examples does not establish a benefit after correcting for
multiple comparisons.  The FGQA and NACH checks illustrate mathematical
design constraints.  Three elements support this case study:

\begin{enumerate}
  \item two budget normalizations, replacement of one model's gain by
  another's, and refitting of the classification layer that expose the difference between robustness comparison and
  module attribution;
  \item elementary scalar-gain and QFI calculations used as explanatory
  tools and design checks, rather than new general robustness theorems; and
  \item a leakage-safe physical threat model, paired ten-seed evaluation,
  and reproducible evidence with \BundleChecks{} independent checks.
\end{enumerate}

We do \emph{not} claim quantum advantage, measured-grid generalization,
hardware robustness, or coverage of a fully adaptive attack battery.  These
boundaries are part of the result, not qualifications deferred to a supplement.

\section{Background and related work}

\paragraph{Quantum attention.}
Hybrid quantum self-attention has used circuit-derived queries, keys, and
values with Gaussian projected attention in sequence classifiers
\citep{li2024quantum}.  Our four-qubit block follows that broad construction,
but we study security-evaluation and module-attribution questions rather than
claiming an architectural or computational advantage.

\paragraph{Quantum adversarial robustness.}
Adversarial examples extend to quantum classifiers in both classical-input and
quantum-state formulations \citep{lu2020quantum,liu2020vulnerability}.
Quantum hypothesis testing supplies a principled connection between prediction
margins, trace distance, and provable robustness \citep{weber2021optimal}.
Noise-based analyses instead derive guarantees from randomized quantum
rotations and differential privacy \citep{huang2023certified}, while recent
many-body results identify encoding-, locality-, and scrambling-dependent
protection against classical adversaries \citep{dowling2026guarantees}.
Separate empirical work has reported regimes in which quantum models show
favorable adversarial behavior \citep{west2023towards}.  Our question is more
diagnostic: when three putative robustness modules are removed one at a time,
which observed quantities actually move under a fixed physical threat model?
Those formal assumptions do not apply automatically to our generated
physical attack or deterministic simulator.  We use QFI as a local
measure of input sensitivity, drawing on its standard pure-state geometry
\citep{stokes2020quantum}, but do not infer global robustness from QFI alone.

\paragraph{Physical FDIA.}
In linearized power-system state estimation, measurements satisfy
$z=Hx+e$.  An injection $a=Hc$ remains in the model column space and can evade
conventional residual-based bad-data detection \citep{liu2009false}.  Allowing
all of $\operatorname{col}(H)$, however, assumes control of every affected
meter and becomes uninformative after a severe dimensionality reduction.  We
instead declare two controllable state directions before evaluation.  This
creates a limited meter set and, crucially, a nontrivial orthogonal complement
in the four-dimensional circuit-input space.

Data-driven FDIA detectors increasingly exploit temporal structure: recurrent
models have been evaluated on dynamic IEEE test systems
\citep{zhang2022dynamic}.  The closest QML precedent aligns a quantum kernel
for sparse FDIA detection and studies simulator noise and kernel scaling
\citep{huang2026aligning}.  Our study asks a complementary question.  We do not
claim superior detection accuracy; we hold a raw-space physical threat model
fixed and audit whether three proposed robustness modules change clean,
adversarial, or uncertainty-aware outcomes under paired knockouts.

\paragraph{Reliable attack evaluation.}
Iterative first-order attacks are meaningful only when gradients, projections,
and threat-set nesting are correct \citep{madry2018towards}.  Simultaneous perturbation stochastic approximation
(SPSA) \citep{spall1992multivariate} probes a different optimization path, while
expectation over transformation (EOT) is a standard diagnostic for stochastic
defenses \citep{athalye2018obfuscated}.  In deterministic statevector simulation,
EOT gradients equal projected gradient descent (PGD) gradients exactly; we therefore reuse the PGD
candidate and do not count EOT as independent corroboration.

\section{Benchmark and physical threat model}
\label{sec:threat}

\subsection{Generated temporal case-14 data}

For each seed we generate $N=200$ balanced six-frame trajectories.  Let
$H\in\mathbb{R}^{34\times 13}$ contain 14 active-power bus-injection rows and
20 directed branch-flow rows for the standard case-14 topology.  A latent
non-slack angle state follows
\begin{equation}
  x_0\sim\mathcal{N}(0,0.35^2I),\qquad
  x_t=0.98x_{t-1}+\eta_t,\quad
  \eta_t\sim\mathcal{N}(0,0.025^2I),
  \label{eq:trajectory}
\end{equation}
and clean baseline measurements are $z_t^{\mathrm{base}}=Hx_t+e_t$ with
$e_t\sim\mathcal{N}(0,0.01^2I)$.  For attacked trajectories an onset is drawn
uniformly from frames 1 through 4.  A signed two-dimensional target coefficient
has coordinate magnitudes sampled uniformly from $[0.10,0.25]$.  Writing this
generator coefficient as $c_t^{(0)}$ (zero before onset), the observed input is
$z_t=z_t^{\mathrm{base}}+A c_t^{(0)}$.  It reaches the sampled target in one
frame and remains there through frame 5.  These are generated operating
trajectories, not supervisory control and data acquisition (SCADA) measurements.

We stratify entire trajectories into 70\% train and 30\% test partitions.
A four-component principal component analysis (PCA) reducer is fitted only to
all training frames.  To expose temporal
change rather than absolute operating state, every reduced frame is expressed
relative to its trajectory's frame zero, then scaled using training-trajectory
contrasts only and mapped to $[-\pi/2,\pi/2]$ at the observed training extrema.
Test and attacked contrasts are not clipped, so the reducer and
Eq.~\eqref{eq:chain} retain the same linear Jacobian.  A separate data-learnability
check generates 1,000 trajectories for each of ten dataset
seeds under the same split and preprocessing rules.  Temporal-statistics
logistic regression reaches $\TemporalGateMean\pm\TemporalGateSD$ held-out
balanced accuracy, whereas logistic regression on flattened contrasts reaches
$\FlatLogisticGateMean\pm\FlatLogisticGateSD$.
Nonlinear flattened controls also learn: a 16-unit multilayer perceptron (MLP) reaches
$\FlatMLPGateMean\pm\FlatMLPGateSD$ and a random forest reaches
$\FlatForestGateMean\pm\FlatForestGateSD$.  Temporal logistic regression, the MLP, and the random forest exceed
\method{}'s $\FullClean\pm\FullCleanSD$ clean accuracy.  Their larger
1,000-trajectory evaluation is neither parameter- nor attack-matched, so they
establish learnability, not superiority of the summary features or architecture.
The quantum detector is a small methodological example, not a competitive
accuracy benchmark.

\subsection{Limited-state attacker}

The attacker controls zero-based non-slack state columns 2 and 3, corresponding
to MATPOWER buses 4 and 5:
\begin{equation}
 A=H_{:,(2,3)}\in\mathbb{R}^{34\times 2},\qquad
 z^{\mathrm{adv}}_0=z_0,\quad
 z^{\mathrm{adv}}_t=z_t+A\delta_t\;(t=1,\ldots,5),
 \quad \|\delta\|_\infty\le\epsilon .
 \label{eq:attack}
\end{equation}
Exactly 16 measurement rows can change and 18 protected rows are identically
zero in $A$.  Coefficients vary by frame, but frame zero is immutable.  If $g$ is the fitted reducer and
$J_g\in\mathbb{R}^{4\times34}$ its Jacobian for an attacked frame, the orthogonal projector onto the attack directions after reduction is
\begin{equation}
  P_A=\operatorname{proj}\!\left(\operatorname{col}(J_gA)\right),
  \qquad \operatorname{rank}(P_A)=2.
  \label{eq:projector}
\end{equation}
Thus both $P_A$ and $I-P_A$ are nonzero.  For loss $\ell$, white-box gradients
are evaluated through the raw physical map,
\begin{equation}
  \nabla_{\delta_t}\ell=A^\top J_g^\top
  \nabla_{g(z^{\mathrm{adv}}_t)}\ell,
  \label{eq:chain}
\end{equation}
and are checked against finite differences.
Circuit-input derivatives use the exact two-term parameter-shift rule for each
single-generator $R_Y$ occurrence \citep{schuld2019evaluating}; reuploaded
occurrences are summed before applying Eq.~\eqref{eq:chain}.
The evaluation attacker has white-box access to the trained model, fitted
reducer, physical Jacobian, labels, and loss; meter support and coefficient
budget, rather than knowledge, limit its capability.

The planned evaluation adds feasible $A\delta_t$ modifications to both classes
while retaining the original trajectory label.  For an already attacked
trajectory, this is a label-preserving evasion test; for an originally clean
trajectory, a nonzero FDIA can change the example from clean to attacked.  We therefore
call the aggregate curve a \emph{physical stress test}, not a complete
class-conditional security guarantee.  A post hoc \emph{positive-only evaluation} includes only trajectories with
an injected attack before evaluation.  It checks that the additional
perturbation does not cancel that attack.  This preserves the generator's
label, not necessarily an operational definition of a harmful attack.

\section{Model and module hypotheses}
\label{sec:model}

\subsection{Quantum attention block}

Each reduced frame is a four-dimensional token, encoded twice on $Q=4$
qubits by Hadamards,
feature-scaled $R_Y$ rotations, and a CZ ring, following the data-reuploading
motif \citep{perezsalinas2020data}.  Three independent one-layer variational
blocks use $R_Y$--$R_Z$ rotations and a controlled-NOT (CNOT) ring.  Their
expectation values
define scalar queries and keys and a four-vector value:
\begin{equation}
 q_t=\langle Z_0\rangle_{U_q},\quad
 k_t=\langle Z_0\rangle_{U_k},\quad
 v_{t,j}=\langle Z_j\rangle_{U_v}.
\end{equation}
Gaussian attention and a linear classification layer (the head) give
\begin{equation}
 \alpha_{ts}=\frac{\exp[-(q_t-k_s)^2/(2\tau^2)]}
 {\sum_u\exp[-(q_t-k_u)^2/(2\tau^2)]},\quad
 h=\frac{1}{6}\sum_t\sum_s\alpha_{ts}v_s,\quad
 p=\operatorname{softmax}(Wh+b).
 \label{eq:attention}
\end{equation}
The attention width starts at $\tau=0.8$.  The model has 24 variational circuit
parameters and 39 trainable parameters in total (including $W,b,\tau$, and four
input gains).  Before device-specific transpilation, each four-qubit query,
key, or value circuit has depth 18 and gate counts
$\{R_Y:12,H:8,CZ:8,R_Z:4,CX:4\}$; tokens are evaluated as separate circuit
bindings rather than entangled across time.

\subsection{Modules and knockouts}

\paragraph{Encoding bandwidth governor (\ebg{}).}
Each input angle becomes $\gamma_jx_j$ with
$10^{-3}\le\gamma_j\le1$ and $\sum_j\gamma_j\le0.4Q=1.6$.  Projection is
applied after each optimizer update.  The encoder sensitivity bound used for budget normalization is
\begin{equation}
 L_{\mathrm{enc}}=r\sum_{j=1}^{Q}|\gamma_j|,
 \label{eq:lipschitz}
\end{equation}
where $r=2$ reuploads.  The $-\ebg$ knockout fixes $\gamma=\mathbf{1}$,
removes it from the trainable vector, and hence fixes
$L_{\mathrm{enc}}=8$.  Our hypothesis is that limiting input bandwidth reduces
empirical sensitivity; because Eq.~\eqref{eq:lipschitz} also appears in the
radius denominator, radius changes are partly structural and must not be
treated as independent confirmation.

\paragraph{Fisher-geometric regularizer (\fgqa{}).}
For the pure-state QFI $F$ of the query circuit with respect to the mean
of the six input tokens, training adds
\begin{equation}
 \mathcal{L}_{\mathrm{FGQA}}=\lambda_{\mathrm{in}}
 \operatorname{tr}(P_AF)-\lambda_{\mathrm{out}}
 \operatorname{tr}[(I-P_A)F],
 \quad (\lambda_{\mathrm{in}},\lambda_{\mathrm{out}})=(0.01,0.001).
 \label{eq:fgqa}
\end{equation}
The full matrix is computed by central statevector differences with step
$10^{-4}$.  During each SPSA loss evaluation, the penalty is evaluated on the
first trajectory in the sampled minibatch.  It is a training regularizer and
measure of input sensitivity, not a hardware-estimated natural gradient.  The
$-\fgqa$ condition sets both contributions to zero.

\begin{proposition}[Downstream-unitary invariance]
\label{prop:qfi-invariance}
Let $|\psi_{\theta,\eta}(x)\rangle=V_\theta
U_\eta(x)|0\rangle$, where $V_\theta$ is unitary and independent of input $x$.
The pure-state QFI with respect to $x$ is independent of $\theta$ and equals
the input QFI of $U_\eta(x)|0\rangle$.  Consequently, a loss that depends only
on this input QFI has zero gradient with respect to parameters appearing
exclusively in $V_\theta$.
\end{proposition}
\begin{proof}
Writing $|\phi_\eta(x)\rangle=U_\eta(x)|0\rangle$ gives
$|\partial_i\psi\rangle=V_\theta|\partial_i\phi\rangle$.  Unitarity therefore
preserves both $\langle\partial_i\psi|\partial_j\psi\rangle$ and
$\langle\partial_i\psi|\psi\rangle$ in the pure-state QFI formula, eliminating
all dependence on $\theta$.
\end{proof}

Here $\Gamma=\operatorname{diag}(\gamma_1,\ldots,\gamma_Q)$,
$V_\theta=V_q(\theta_q)$, and $U_\eta(x)=U_{\rm enc}(\Gamma x)$.  The penalty can therefore move the
governor $\Gamma$ but is exactly independent of the downstream query angles
$\theta_q$ and therefore cannot train them.  To train a circuit parameter through input QFI,
that parameter must affect the encoding rather than occur only in a subsequent
input-independent unitary.  This follows directly from unitary invariance of
pure-state QFI.  Here, the regularizer can train input scales but not query
angles.

We summarize an attack curve by its area under robust accuracy (AURA),
normalized by the evaluated budget range.

\begin{proposition}[Axis dependence of scalar-gain AURA]
\label{prop:gain-identifiability}
Let a pre-encoder module apply scalar gain $g>0$.  For an arm with gain $g$,
let $R_g(\delta)$ denote robust accuracy as a function of downstream
perturbation displacement $\delta$.  On a common physical range $[0,E]$,
\begin{equation}
 A_{\rm phys}(g;E)=\frac{1}{E}\int_0^E R_g(g\epsilon)d\epsilon
 =\frac{1}{gE}\int_0^{gE}R_g(\delta)d\delta .
 \label{eq:gain-aura}
\end{equation}
If all models share the same response curve, $R_g=R$, and $R$ is nonincreasing, then
$g_1>g_0$ implies $A_{\rm phys}(g_1;E)\le A_{\rm phys}(g_0;E)$, strictly when
$R$ decreases on a set of positive measure within $[0,g_1E]$.
On a common displacement range
$[0,D]$, using $E_g=D/g$ removes this mechanical gain dependence.  If response curves differ, either normalization also reflects
changes in the model's response; one comparison cannot isolate the gain effect.
\end{proposition}
\begin{proof}
The first equality is the physical-axis definition and the second substitutes
$\delta=g\epsilon$.  Under $R_g=R$, the average of a nonincreasing function
over $[0,u]$ is nonincreasing in $u$, which gives the physical-axis ordering.
Substituting $E_g=D/g$ yields $D^{-1}\int_0^D R(\delta)d\delta$ for every $g$.
\end{proof}

This change-of-variables argument makes explicit the familiar dependence of
input sensitivity on scale; it is not specific to quantum models.
Proposition~\ref{prop:gain-identifiability} predicts a mechanical advantage
for lower gain and its removal after displacement matching only when the
models share the same response curve.  Neither more examples nor more
seeds repairs a failure of that condition.  EBG uses a diagonal $\Gamma$, not
a scalar $g$, and $L_{\rm enc}$ is an upper-bound normalizer, so the proposition
explains the role of scaling without exactly reparameterizing our detector.
We compare both budget normalizations, then replace the gains to check whether
the model still makes the same predictions without additional perturbations.

\paragraph{Finite-sample radius procedure (\nach{}).}
For $n$ repeated predictions, simultaneous one-sided Clopper--Pearson bounds
\citep{clopper1934use} give $\underline p_A$ and $\overline p_B$ for the two
most frequent classes.  The implementation abstains if
$\underline p_A\le\overline p_B$ and otherwise reports
\begin{equation}
 \Rdiag=\frac{\underline p_A-\overline p_B}{2L_{\mathrm{enc}}}.
 \label{eq:radius}
\end{equation}
The familywise level is $10^{-3}$ with Bonferroni correction across ten
evaluated examples.  The $-\nach$ comparator substitutes point frequencies
for confidence bounds and is marked invalid by construction.  Importantly,
the simulator is deterministic and uses no sampled quantum measurements: repeated calls
return the same prediction.  Equation~\eqref{eq:radius} is therefore a
finite-sample implementation diagnostic in this study, not evidence of a
characterized device-noise channel or an end-to-end hardware certificate.

\section{Experimental protocol}
\label{sec:protocol}

\paragraph{Paired design.}
We use seeds 3000--3009 and five conditions: full, $-\ebg$, $-\fgqa$,
$-\nach$, and all modules off.  Within a seed, data, split, initialization
streams, evaluation subset, and attack streams are deterministically derived
from the same root seed.  Clean accuracy uses all 60 held-out trajectories before additional
evaluation perturbations.  Attacks use a deterministic approximately balanced subset of 20;
the radius diagnostic uses its first ten examples.  Model selection is not
performed on test outcomes.

\paragraph{Training.}
All runs use Qiskit's precision-zero statevector estimator
\citep{javadiabhari2024qiskit}.  Twenty hybrid SPSA steps
\citep{spall1992multivariate} update circuit parameters, temperature, and
governor variables with batch size 8, $a=0.01$, $c=0.05$, and decays
$t^{-0.602}$ and $t^{-0.101}$.  The linear head receives analytic minibatch
gradients at learning rate 0.05.  Before and after SPSA, it is fitted for 500
Adam steps on cached training-only quantum features (learning rate 0.05,
$\ell_2=10^{-4}$).  This deliberately small-scale study is designed for
module attribution rather than scaling claims.

\paragraph{Attacks and normalized accuracy curves.}
At $\epsilon\in\{0,0.0375,0.075,0.1125,0.15\}$ we run ten-step,
zero-initialized sign PGD and projected SPSA in the additional-coefficient tensor
$\delta\in\mathbb{R}^{5\times2}$.  PGD uses step size $2.5\epsilon/10$.  SPSA uses
perturbation $0.05t^{-0.101}$ and update magnitude $0.02t^{-0.602}$.  For every
sample and budget, the candidate with
maximum cross-entropy is selected.  The strongest previously evaluated
smaller-budget candidate is retained because the threat sets are nested.  EOT
is listed in the manifest but is algebraically identical to PGD on this deterministic
simulator and is reused.  We compute AURA by trapezoidal integration
and divide by the evaluated budget range, 0.15 for this test.  Neither optimizer uses randomized restarts.

\paragraph{Checks for gradient masking.}
Uninformative gradients can make attacks appear weaker than they are.
For each seed and condition, we compare PGD with SPSA and random perturbations,
test a larger attack budget, and check that retaining earlier attacks prevents
loss from decreasing as the budget grows.  We check one example per condition;
details and separate optimizer paths are in Appendix~\ref{app:curves}.

\paragraph{Inference.}
We compare each knockout with the full model using differences within each seed.  We report
means and sample standard deviations, 95\% bias-corrected and accelerated
(BCa) bootstrap intervals with 10,000 paired resamples
\citep{efron1987better}, exact two-sided Wilcoxon signed-rank tests
\citep{wilcoxon1945individual}, and Cliff's delta
\citep{cliff1993dominance}.  Tied
nonzero ranks are handled by exhaustive sign permutation rather than a library
``exact'' mode that assumes no ties.  Holm--Bonferroni correction
\citep{holm1979simple} covers all 12 reported condition--metric contrasts.
Invalid point-radius comparisons remain in the multiplicity family but cannot
support a certification claim.

\paragraph{Post hoc positive-only evaluation.}
After inspecting the initial result, we specified a follow-up using the same
20 already attacked test trajectories for every condition within each seed.
We rerun the unchanged attacks and report the positive-only correct-class rate
(the fraction still classified as attacked) and its normalized area.
The comparisons form a separate 12-contrast Holm family.  At each budget,
the maximum absolute combined coefficient across nonreference frames must
exceed $10^{-8}$, checking that the evaluation attack does not cancel the
generator's injection.  This is not evidence of label validity in an
operational grid.

\paragraph{Controls for input scaling.}
We next evaluated all 30 attacked examples per seed using the original
physical budgets and rescaled budgets.  The primary rescaling uses the
mean full-model bound divided by the no-EBG bound,
$3.1349978/8=\MatchedGlobalScale$; a second analysis uses each seed's ratio.
We then tested gain replacement: copy the full model's learned $\Gamma$ into
the paired trained $-\ebg$ model, keeping every other parameter fixed.
The AURA difference decomposes as
$A_N-A_F=(A_N-A_T)+(A_T-A_F)$, where $F$, $N$, and $T$ denote the full,
no-EBG, and gain-replaced models.  The two components form a Holm family.
After gain replacement disrupted detection, we refitted only $(W,b)$ on
the 140 training trajectories (500 Adam steps, learning rate 0.05,
$\ell_2=10^{-4}$), keeping the circuit, temperature, reducer, and gains fixed.
This comparison is descriptive because retraining changes the decision rule.
A 39-parameter classical tanh attention model provides a separate comparison
with matched parameter count and evaluation settings; it received no tuning.

Each follow-up was specified before its own execution, after observing earlier
results.  This limits implementation changes but does not make the sequence
preregistered.  Separate Holm families do not correct for selecting follow-up
questions after seeing earlier results.  These controls investigate the
source of the differences rather than independently confirm a robustness benefit.

\section{Results}
\label{sec:results}

The ordering of the full and no-EBG models reverses when the budget
normalization changes.  We first report the physical stress test, then
compare the rescaled budgets and gain replacements.  The final two
subsections explain the regularizer and radius design checks.

\begin{table}[t]
  \caption{Ten-seed endpoints (mean $\pm$ sample SD).  AURA uses the
  20-trajectory stress-test subset; clean accuracy uses the full test split.
  $^\dagger$Point radii omit confidence bounds and are not certificates.}
  \label{tab:endpoints}
  \centering
  \small
  \setlength{\tabcolsep}{4.5pt}
  \begin{tabular}{lcccc}
    \toprule
    Condition & Clean accuracy & AURA & Mean $\Rdiag$ & Procedure\\
    \midrule
    \EndpointRows
    \bottomrule
  \end{tabular}
\end{table}

\subsection{Changing budget normalization reverses the comparison}
\label{sec:normalization-failure}

\paragraph{Common physical attack budget.}

The full model obtains clean accuracy $\FullClean\pm\FullCleanSD$ and AURA
$\FullAURA\pm\FullAURASD$ (Table~\ref{tab:endpoints}).  Removing \ebg{} changes
clean accuracy by $\EBGCleanDelta$ with BCa interval
$[\EBGCleanCILow,\EBGCleanCIHigh]$ and Holm $p=\EBGCleanHolmP$; clean
performance is therefore not detectably different.  In contrast, $-\ebg$
reduces AURA by $\EBGAURAReduction$ (BCa
$[\EBGAURACILow,\EBGAURACIHigh]$, Holm $p=\EBGAURAHolmP$).
Figure~\ref{fig:robust} shows separation at every positive budget.  At
$\epsilon=0.075$, mean robust accuracy is $\FullAtMid$ for full and
$\NoEBGAtMid$ for $-\ebg$; at $0.15$ it is $\FullAtMax$ and
$\NoEBGAtMax$, respectively.  Both models perform poorly
at the largest budget: the difference is only $0.015$ and neither exceeds
$0.065$ accuracy.  The paired separation is statistically significant for this
stress-test comparison; it is not by itself evidence of an intrinsic
robustness mechanism.

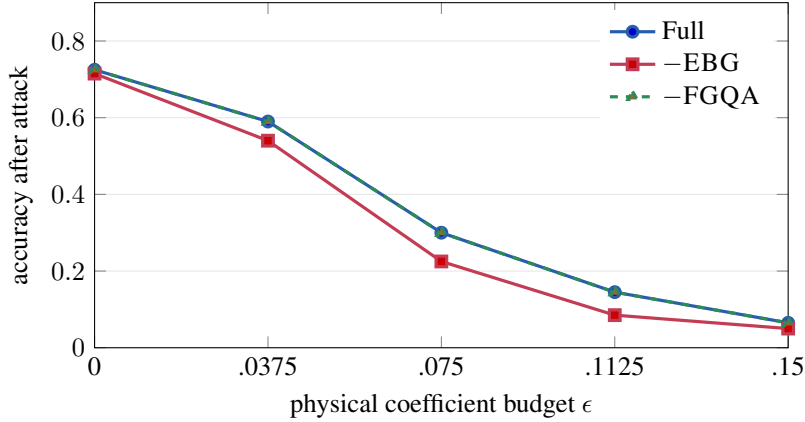
\begin{figure}[t]
  \centering
  \begin{tikzpicture}
    \begin{axis}[
      width=0.77\linewidth,height=0.44\linewidth,
      xlabel={physical coefficient budget $\epsilon$},
      ylabel={accuracy after attack},
      xmin=0,xmax=0.15,ymin=0,ymax=0.9,
      xtick={0,0.0375,0.075,0.1125,0.15},
      xticklabels={0,.0375,.075,.1125,.15},
      ymajorgrids=true,grid style={gray!20},
      legend style={at={(0.98,0.98)},anchor=north east,draw=none,fill=white},
      legend cell align={left},
      every axis plot/.append style={very thick,mark size=2pt},
    ]
      \addplot+[color=fullblue,mark=*] table[x=eps,y=full] {generated/robust_curves.dat};
      \addlegendentry{Full}
      \addplot+[color=ebgred,mark=square*] table[x=eps,y=no_ebg] {generated/robust_curves.dat};
      \addlegendentry{$-\ebg$}
      \addplot+[color=fgqagreen,mark=triangle*,dashed] table[x=eps,y=no_fgqa] {generated/robust_curves.dat};
      \addlegendentry{$-\fgqa$}
    \end{axis}
  \end{tikzpicture}
  \caption{Mean stress-test accuracy over ten paired seeds.  Full
  and $-\fgqa$ coincide at recorded precision.  Curves take the worst loss over
  implemented PGD/SPSA candidates and all previously evaluated smaller-budget
  candidates.  Standard deviations are in Appendix~\ref{app:curves}.}
  \label{fig:robust}
\end{figure}

All \NumRows{} checks for gradient masking pass.  PGD supplies
\PGDWinners{} of \TotalWinners{} selected attacks and SPSA supplies
\SPSAWinners{}.  Retaining earlier attacks makes the reported loss monotone;
separate optimizer paths remain available in Appendix~\ref{app:curves}.

\paragraph{Matching the encoder sensitivity bound.}

The original radius contrast is coupled to Eq.~\eqref{eq:lipschitz}: the
learned governor has mean
$\sum_j\gamma_j=\FullGammaSum\pm\FullGammaSumSD$ and
$L_{\mathrm{enc}}=\FullLipschitz\pm\FullLipschitzSD$, whereas $-\ebg$ fixes
these values to 4 and 8.  At fixed physical $\epsilon$, the no-EBG encoder
has a sensitivity bound about $2.55\times$ larger.  Against
the cap $\sum_j\gamma_j\le1.6$, the learned sum is nearly saturated; at this
scale EBG primarily redistributes an aggregate gain already pinned close to
0.4 per angle rather than learning a substantially smaller bandwidth.

The budget-rescaling control examines this difference (Figure~\ref{fig:matched}).
On all 30 positive examples per seed, full-model AURA is
$\ControlQuantumAURA\pm\ControlQuantumAURASD$.  The unmatched no-EBG-minus-full
difference is $\MatchedUnmatchedDelta$ (exact $p=\MatchedUnmatchedRawP$),
while the global match reverses the ordering:
$\MatchedGlobalDelta$ (BCa [$\MatchedGlobalCILow,\MatchedGlobalCIHigh$],
exact $p=\MatchedGlobalRawP$).  Rescaling by each seed's bound gives a similar result
(Appendix~\ref{app:confound-controls}).  Matching changes the physical threat
set: the no-EBG endpoint of 0.417 uses $\epsilon\approx0.0588$, versus the full
arm's $0.15$, and each AURA is normalized by its own physical range.
A common physical grid therefore compares unequal encoder-displacement bounds;
a matched bound compares unequal physical budgets; neither budget normalization alone isolates
EBG's effect.  Moreover, the global bound does not match realized
sensitivity within the rank-2 attack subspace.  A subspace-restricted operator
norm is a natural further control, not an evaluated normalization here.

The paired gain replacement changes the model's predictions even without
additional perturbations, violating Proposition~\ref{prop:gain-identifiability}'s
shared-response assumption.  Replacing only the trained
no-EBG model's unit gain with full-model $\Gamma$ reduces the unperturbed detection rate from $0.663$ to $0.047$; AURA after replacement is
$\GainTransplantAURA\pm\GainTransplantAURASD$.  The exact identity
$N-F=(N-T)+(T-F)$ still holds, but its large opposing components describe a
change in the decision function (Appendix~\ref{app:confound-controls}).
Because the change is present at $\epsilon=0$, it cannot isolate
robustness due to reduced input scale or show whether adjusting the head
alone would restore detection.

Refitting only the classification head restores the unperturbed detection
rate to $\HeadRefitPositiveClean$, with the circuit, temperature, reducer,
and gains unchanged.  On the same physical budgets, AURA exceeds the original
no-EBG model's by $\HeadRefitNoEBGAURADelta$ (BCa
[$\HeadRefitNoEBGAURACILow,\HeadRefitNoEBGAURACIHigh$]).
This comparison is descriptive: gain and head both change.  Mean prediction
agreement with the original no-EBG model is $\HeadRefitAgreement$ across
the full 60-example test set before additional perturbations.
Appendix~\ref{app:confound-controls} gives the test, effect size, and checks.

An operational reading is that physical state-angle coefficients are the
relevant constraint, making attenuation useful.  The repaired comparison is
consistent with that possibility, but its descriptive exact test
($p=\HeadRefitNoEBGAURARawP$) and the original positive-only test
($p=\PositiveEBGAURARawP$) do not corroborate that deployment claim.
At ten paired seeds, this is limited evidence, not an estimate of zero effect.

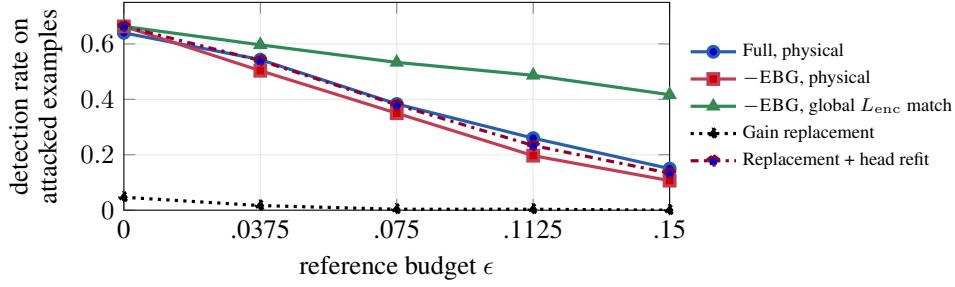
\begin{figure}[t]
  \centering
  \begin{tikzpicture}
    \begin{axis}[
      width=0.63\linewidth,height=0.31\linewidth,
      xlabel={reference budget $\epsilon$},
      ylabel={detection rate on\\attacked examples},
      ylabel style={align=center},
      xmin=0,xmax=0.15,ymin=0,ymax=0.75,
      xtick={0,0.0375,0.075,0.1125,0.15},
      xticklabels={0,.0375,.075,.1125,.15},
      grid=major,grid style={gray!18},
      legend style={at={(1.02,0.5)},anchor=west,draw=none,font=\scriptsize},
      legend cell align={left},
      every axis plot/.append style={very thick,mark size=2pt}]
      \addplot+[color=fullblue,mark=*]
        table[x=eps,y=full] {generated/sensitivity_control_curves.dat};
      \addlegendentry{Full, physical}
      \addplot+[color=ebgred,mark=square*]
        table[x=eps,y=no_ebg_unmatched] {generated/sensitivity_control_curves.dat};
      \addlegendentry{$-\ebg$, physical}
      \addplot+[color=fgqagreen,mark=triangle*]
        table[x=eps,y=no_ebg_global_match] {generated/gain_transplant_curves.dat};
      \addlegendentry{$-\ebg$, global $L_{\rm enc}$ match}
      \addplot+[color=black,mark=diamond*,dotted]
        table[x=eps,y=transplant] {generated/gain_transplant_curves.dat};
      \addlegendentry{Gain replacement}
      \addplot+[color=purple!75!black,mark=pentagon*,dashdotted]
        table[x=eps,y=head_refit] {generated/head_refit_curve.dat};
      \addlegendentry{Replacement + head refit}
    \end{axis}
  \end{tikzpicture}
  \caption{Controls on all 30 attacked test examples.  Matched $-\ebg$ uses physical
  budget $\MatchedGlobalScale\epsilon$; other curves use $\epsilon$.
  Gain replacement disrupts detection; refitting restores the
  unperturbed detection rate but changes both gain and classification head.}
  \label{fig:matched}
\end{figure}

\paragraph{Label-preserving check.}

Removing \ebg{} changes positive-only AURA by $\PositiveEBGAURADelta$
(exact $p=\PositiveEBGAURARawP$, Holm $p=\PositiveEBGAURAHolmP$).
The signed-rank test does not reject; a descriptive negative BCa interval does not override
it.  Restricting evaluation to attacked examples therefore does not establish
an EBG benefit or equivalence between the models (Appendix~\ref{app:positive}).

\paragraph{Classical attention comparison.}
The untuned 39-parameter classical model has classical-minus-quantum clean/AURA gaps
$\ClassicalCleanDelta$/$\ClassicalAURADelta$ (exact $p=\ClassicalCleanRawP$/
$\ClassicalAURARawP$); this single untuned comparison does not establish quantum advantage.
The stronger temporal-logistic baseline already outperforms the quantum model
on clean prediction (Appendix~\ref{app:confound-controls}).

\subsection{Design lesson: check which parameters a regularizer can train}

Full and $-\fgqa$ have identical clean accuracy and AURA at recorded precision
for every seed; their radius difference $\FGQARadiusDeltaSci$ misses the
Holm-adjusted test (Holm $p=\FGQARadiusHolmP$).  Proposition~\ref{prop:qfi-invariance}
shows that the penalty can affect only the four input gains.  With those gains
fixed $-\ebg$ equals all-off.  Full and $-\fgqa$ mean gain sums and both QFI
trace terms in Eq.~\eqref{eq:fgqa} differ by less than $0.001$, consistent with the limited effect on input scaling.
The design lesson follows from the QFI calculation before any experiment:
check which parameters the regularizer depends on before interpreting a knockout.  These runs check
the implementation's behavior; they do not establish a general failure of QFI
regularization or exclude metrics on trainable feature maps.

\subsection{Design lesson: preserve the meaning of a radius}

The $-\nach$ point radius, $\NoNACHRadius\pm\NoNACHRadiusSD$, exceeds the full
diagnostic only because confidence margins were deleted; its artifacts set
\texttt{valid\_certificate=false}.  All 100 deterministic simulator calls are identical,
so the confidence term is constant and the radius varies only through
$L_{\rm enc}$.  Both invalid comparisons remain in the original test family.  The larger
number is a predictable consequence of changing the formula, not an empirical
improvement in certified robustness.

\section{Discussion: failure modes and remedies}

\paragraph{What the sign reversal means.}
A physical budget fixes the measurement-model intervention; an
encoder-displacement bound compensates for input scale.  These constrain
different threat sets.  The scalar calculation explains gain dependence;
diagonal gains and a global bound limit its application here.
Refitting the head restores mean unperturbed detection, but not every prediction.
Equal initial accuracy is necessary for the shared response curves in
Proposition~\ref{prop:gain-identifiability}, but establishes neither their
equality nor preservation of the decision function.

\paragraph{Relation to existing work.}
Transformation-aware classical evaluation \citep{athalye2018obfuscated}
likewise includes the input map in the threat model.  Our attacks differentiate
through the gain and reducer; the problem is attribution, not hidden gradients.
Unlike favorable quantum adversarial results
\citep{west2023towards} or FDIA detection with quantum kernels
\citep{huang2026aligning}, this study establishes no competitive quantum defense.
Its study-specific modules do not demonstrate a sign reversal in a published
QML robustness result.  Checking budgets and prediction preservation in an
existing defense is a concrete next test.

\paragraph{Practical evaluation rule.}
For gain-changing ablations, report both budget normalizations and replace
gains with other parameters fixed.  If unperturbed predictions change,
refit only the classification head and report remaining prediction disagreement.
Changing both gain and head gives a descriptive comparison; more samples
cannot make it a causal estimate of the gain's contribution.

\section{Limitations and broader impact}

Four qubits, ten seeds, and 20-example attacks limit generalization.
Exact tests are coarse at this seed count, and the 12-contrast correction
limits power; failure to reject is not equivalence.
The 30-positive follow-ups remain diagnostic sensitivity analyses.  Global
$L_{\rm enc}$ misses rank-2-restricted gain, and the PGD/SPSA evaluation has no
randomized restarts or broader adaptive coverage.  Twenty training SPSA steps
and one-sample FGQA may also undertrain the model.  Stronger attacks could
lower the reported accuracies; passing the implemented checks is not evidence
of worst-case robustness.

Generated linear-DC data omit AC physics, distribution shift, noise, and
calibrated physical budgets.  Classical baselines outperform the quantum
detector on clean prediction, and we make no hardware or deployment claim.
The intended benefit is more reliable evaluation.  The FDIA code is dual use;
the artifact is limited to a public toy topology and generated measurements,
without operational grid data or a control interface.

\paragraph{Conclusion.}

Changing input gain changes how physical perturbations reach a model.
Here, robustness ordering reverses under two normalizations; neither ordering
alone attributes a module's effect.  Gain replacement and head refitting
restore mean detection but change the decision rule.  Evaluations should
report the fixed budget, check labels and unperturbed predictions, and identify
diagnostic follow-ups.  The QFI and radius checks also show why a module's
mathematical role must precede interpretation of its knockout.  This small
synthetic study supports a methodological warning; a predeclared evaluation
of a published defense with stronger attacks is the next step.

\begingroup
\hyphenpenalty=10000
\begin{ack}
The computation for this work was performed on the high performance computing infrastructure provided by Research Support Services at the University of Missouri, Columbia MO. DOI: 10.32469/10355/97710
\end{ack}
\par
\endgroup

\clearpage
\bibliographystyle{plainnat}
\bibliography{references}

\appendix

\section{Complete experimental specification}
\label{app:spec}

\begin{table}[h]
  \caption{Locked production configuration.}
  \centering
  \small
  \begin{tabular}{ll@{\qquad}ll}
    \toprule
    Quantity & Value & Quantity & Value\\
    \midrule
    Root seeds & 3000--3009 & Conditions & 5 paired\\
    Trajectories/seed & 200 & Train/test & 140/60\\
    Positive fraction & 0.5 & Tokens & 6\\
    Qubits/features & 4 & Reuploads & 2\\
    Variational layers & 1 & Entangler & ring\\
    Initial attention $\tau$ & 0.8 & EBG budget & 1.6\\
    Train steps/batch & 20/8 & Optimizer & hybrid SPSA\\
    SPSA $(a,c)$ & $(0.01,0.05)$ & Classical LR & 0.05\\
    Head fit & 500 steps & Head $\ell_2$ & $10^{-4}$\\
    FGQA weights & 0.01/0.001 & Attack steps & 10\\
    Attack subset & 20 & Budgets & 5, max 0.15\\
    Radius subset & 10 & Repetitions & 100\\
    Familywise $\alpha$ & 0.001 & Radius grid & 6, max 0.05\\
    Executor & exact statevector & Shot count & 0\\
    \bottomrule
  \end{tabular}
\end{table}

For the data-learnability check, temporal summary features concatenate per-frame
$\ell_2$ norms and maximum absolute values of both adjacent-frame changes and
changes from frame zero.  Logistic models use 2,000 maximum iterations; the
flattened MLP has one 16-unit hidden layer, L-BFGS, and 2,000 maximum
iterations; the random forest uses 300 trees and minimum leaf size five.  Each
model's random state equals the dataset seed; other scikit-learn settings are
defaults for version 1.7.2.

Production ran Python 3.13.3 with NumPy 2.2.6, SciPy 1.15.3,
scikit-learn 1.7.2, Qiskit 2.5.2, and Qiskit Aer 0.17.2.
Ten independent Slurm array tasks completed in 75:20--78:05 wall time.  Each
script requested two logical CPUs and 32 GB memory, although the exclusive-node
partition allocated and accounted a full 256-CPU node per task.  Production
used three AMD x86-64 nodes, each configured as two 64-core sockets with two
hardware threads per core and approximately 490 GiB physical memory.  Observed
batch maximum resident memory was 115--127 MiB.  Each task ran one Python
process; the accounting record is not evidence of 256-way scaling.
The post hoc positive-only array used the same ten seeds, nodes, two-CPU and
32-GB requests, and one-process execution pattern.  Its tasks completed in
76:30--79:03 with 115--124 MiB observed batch maximum resident memory.
The matched-Lipschitz array used the same request and environment; tasks
completed in 1:38:19--1:41:14 with 113--127 MiB observed maximum resident
memory.  Its longer runtime reflects four complete 30-example curves per seed.
The gain-replacement array used the same ten seeds, software environment,
two-CPU/32-GB request, and one-process pattern; tasks completed in
25:59--27:21 with 112--119 MiB observed maximum resident memory.
The classical control ran all ten seeds serially in one CPU-only Python
process, made no circuit or shot calls, and recorded 1.071 seconds of training
plus 6.240 seconds of attack evaluation in total.  Peak-memory accounting was
not retained for that lightweight run.

\section{Endpoint and contrast details}
\label{app:contrasts}

\begin{table}[h]
  \caption{All paired contrasts.  Delta is condition minus full.  Holm
  correction covers all 12 condition--metric tests.  $^\ast$AURA ordering for
  gain-changing conditions is normalization-dependent; see
  Section~\ref{sec:normalization-failure}.  $^\dagger$The radius comparator is invalid and
  cannot support a certificate claim.}
  \label{tab:contrasts}
  \centering
  \scriptsize
  \setlength{\tabcolsep}{3.2pt}
  \begin{tabular}{llrrrr}
    \toprule
    Knockout & Metric & Mean $\Delta$ & 95\% BCa CI & Holm $p$ & Cliff $\delta$\\
    \midrule
    \ContrastRows
    \bottomrule
  \end{tabular}
\end{table}

Cliff's delta is reported as a distributional effect-size summary; the paired
mean difference and exact signed-rank test are the primary estimand and test.
For tied nonzero absolute differences, all $2^m$ sign assignments are
enumerated.  For example, the $-\ebg$ clean contrast has exhaustive two-sided
$p=124/512=0.2421875$, rather than the nominal 0.25 returned by an ``exact''
implementation that does not support tied ranks.  This correction does not
change any Holm decision.

\section{Robust-curve dispersion and masking audit}
\label{app:curves}

\paragraph{Masking diagnostics.}
One preselected example per seed and condition is checked at
$\epsilon=0.15$: a high-budget PGD diagnostic ($\epsilon=3$, 40 steps) must
fool; SPSA loss must not
exceed PGD loss by $10^{-6}$; the reported nested-budget loss envelope must be
monotone; and PGD loss must exceed a random feasible perturbation.  A failure
blocks a robustness claim for that row.  Independent, non-enveloped PGD paths
are retained to reveal optimizer nonmonotonicity rather than silently
smoothing it.

\begin{table}[h]
  \caption{Accuracy after attack by budget, mean $\pm$ sample SD over ten
  seeds.}
  \centering
  \small
  \begin{tabular}{lccccc}
    \toprule
    Condition & 0 & 0.0375 & 0.075 & 0.1125 & 0.15\\
    \midrule
    \RobustCurveRows
    \bottomrule
  \end{tabular}
\end{table}

The four masking predicates pass in every one of 50 evaluated rows.  The
underlying independent PGD loss paths are monotone in \FullMonotone/10 full,
\NoEBGMonotone/10 $-\ebg$, \NoFGQAMonotone/10 $-\fgqa$,
\NoNACHMonotone/10 $-\nach$, and \NoAllMonotone/10 all-off rows.  The
$-\nach$ paths reuse the identical full-model evaluation, so the distinct
total is \IndependentMonotoneTotal/\DistinctAttackEvaluationRows{}.
The reportable
curve uses the exact nested-set envelope: if $\epsilon_1<\epsilon_2$, every
candidate feasible at $\epsilon_1$ remains feasible at $\epsilon_2$.  This
operation is attack bookkeeping, not post hoc curve smoothing.
Across distinct evaluations, PGD supplies \PGDWinners{} winners, SPSA
\SPSAWinners{}, and reused EOT \EOTWinners{}; \EvaluationReuseRows{} rows are
certificate-only reuse.  The \IndependentNonMonotoneTotal{} independent-path
deviations remain explicit in the audit record.

\section{Post hoc positive-only details}
\label{app:positive}

The rates below use the deterministic 20-example generator-positive subset;
they are not full-test clean accuracy.  The $\epsilon=0$ column is the
unperturbed correct-class rate on that subset.  Radius rows are retained for
complete condition accounting, with point-radius invalidity unchanged.
The rerun contains \PositiveDistinctAttackEvaluationRows{} distinct attack
evaluations and \PositiveEvaluationReuseRows{} certificate-only reuse rows.
Its \PositiveSemanticPasses{} checks that the combined attack is nonzero comprise
\PositiveDistinctSemanticPasses{} distinct checks
and \PositiveSemanticReuseChecks{} reuse checks; the minimum retained signal is
\PositiveMinSignal{}.  Removing EBG lowers AURA in
\PositiveEBGAURALowerSeeds{} seeds, ties in \PositiveEBGAURAEqualSeeds{}, and
raises it in \PositiveEBGAURAHigherSeeds{}.  At $\epsilon=0.075$, full and
$-\ebg$ rates are \PositiveFullAtMid{} and \PositiveNoEBGAtMid{}; at $0.15$
they are \PositiveFullAtMax{} and \PositiveNoEBGAtMax{}.

\begin{table}[h]
  \caption{Positive-only endpoints (mean $\pm$ sample SD).}
  \centering
  \small
  \setlength{\tabcolsep}{5pt}
  \begin{tabular}{lccc}
    \toprule
    Condition & Positive AURA & $\epsilon=0$ rate & $\epsilon=0.15$ rate\\
    \midrule
    \PositiveEndpointRows
    \bottomrule
  \end{tabular}
\end{table}

\begin{table}[h]
  \caption{Positive-only correct-class rate by budget (mean $\pm$ sample
  SD).}
  \centering
  \small
  \setlength{\tabcolsep}{4pt}
  \begin{tabular}{lccccc}
    \toprule
    Condition & 0 & 0.0375 & 0.075 & 0.1125 & 0.15\\
    \midrule
    \PositiveCurveRows
    \bottomrule
  \end{tabular}
\end{table}

\begin{table}[h]
  \caption{Post hoc positive-only contrasts.  Holm correction is separate
  over all 12 rows.  $^\dagger$Point-radius comparator is invalid.}
  \centering
  \scriptsize
  \setlength{\tabcolsep}{3.2pt}
  \begin{tabular}{llrrrr}
    \toprule
    Knockout & Metric & Mean $\Delta$ & 95\% BCa CI & Holm $p$ & Cliff $\delta$\\
    \midrule
    \PositiveContrastRows
    \bottomrule
  \end{tabular}
\end{table}

For the positive-only EBG AURA contrast, Cliff's descriptive effect estimate
is $\PositiveEBGAURACliff$; it does not change the signed-rank decision.

\section{Input-scaling control details}
\label{app:confound-controls}

The matched-$L_{\mathrm{enc}}$ control uses every positive test example, so
each curve ordinate lies on a 1/30 grid.  AURA is normalized by each arm's
actual physical epsilon range; the reference budget aligns the encoder-bound
product, not the raw threat sets.  Global matching was the primary
comparison specified before this run; per-seed matching checks its sensitivity.

\begin{table}[h]
  \caption{Matched-sensitivity endpoints over ten seeds.  The last column is
  the mean maximum physical coefficient budget; matched arms intentionally use
  a smaller raw budget.}
  \centering
  \small
  \begin{tabular}{lccc}
    \toprule
    Arm & Positive AURA & Final rate & Physical $\epsilon_{\max}$\\
    \midrule
    \MatchedControlRows
    \bottomrule
  \end{tabular}
\end{table}

The unmatched, global-match, and paired-match no-EBG-minus-full AURA
differences are $\MatchedUnmatchedDelta$, $\MatchedGlobalDelta$, and
$\MatchedPairedDelta$, respectively.  The corresponding exact p-values are
$\MatchedUnmatchedRawP$, $\MatchedGlobalRawP$, and
$\MatchedPairedRawP$.  The global-match BCa interval is
$[\MatchedGlobalCILow,\MatchedGlobalCIHigh]$ with Cliff's
$\delta=\MatchedGlobalCliff$; the paired-match interval is
$[\MatchedPairedCILow,\MatchedPairedCIHigh]$.
\MatchedControlDecision{} All 40 rows use the identical
complete-positive index set within seed, reproduce accepted model/training
fields, pass masking diagnostics, and preserve all 6,000 combined
generator-plus-evaluation FDIA signals.

The evaluation-only gain replacement uses the same complete-positive subset and
common physical grid.  Table~\ref{tab:gain-transplant} reports its exact paired
decomposition.  The total row is accepted context and is not included in the
new two-component Holm family.  Because gain replacement changes the zero-budget
correct-class rate from $0.663$ to $0.047$, the two components quantify a
change in the decision function, not an attenuation-only causal effect. The components are
$N-T=\GainComponentDelta$ and $T-F=\GainResidualDelta$, both with
Holm $p=\GainComponentHolmP$ and $|\delta|=1$; the table retains their
BCa intervals.

\begin{table}[h]
  \caption{Gain-replacement AURA decomposition over ten paired seeds.
  $N$ is trained no-EBG, $F$ is full, and $T$ inserts only full $\Gamma$ into
  $N$.  The identity $N-F=(N-T)+(T-F)$ reconstructs with zero error.}
  \label{tab:gain-transplant}
  \centering
  \scriptsize
  \setlength{\tabcolsep}{3.2pt}
  \begin{tabular}{lrrrrr}
    \toprule
    Component & Mean $\Delta$ & 95\% BCa CI & Exact $p$ & Holm $p$ & Cliff $\delta$\\
    \midrule
    \GainTransplantRows
    \bottomrule
  \end{tabular}
\end{table}

The head-refit diagnostic holds the circuit and replaced gain fixed and
changes only $(W,b)$.  Its unperturbed detection rate is
$\HeadRefitPositiveClean\pm\HeadRefitPositiveCleanSD$, full-test clean accuracy
is $\HeadRefitTestClean\pm\HeadRefitTestCleanSD$, and AURA is
$\HeadRefitAURA\pm\HeadRefitAURASD$.  Refitting the head changes AURA by
$\HeadRefitFrozenAURADelta$ (BCa
[$\HeadRefitFrozenAURACILow,\HeadRefitFrozenAURACIHigh$]); the control is
descriptive and has no newly adjusted family.  Circuit parameters, temperature,
and $\Gamma$ retain identical byte hashes before and after every refit.
Against the original recipient, the repaired arm matches the mean unperturbed detection
rate (difference $\HeadRefitNoEBGCleanDelta$, BCa
[$\HeadRefitNoEBGCleanCILow,\HeadRefitNoEBGCleanCIHigh$]) and changes AURA by
$\HeadRefitNoEBGAURADelta$ (BCa
[$\HeadRefitNoEBGAURACILow,\HeadRefitNoEBGAURACIHigh$], exact
$p=\HeadRefitNoEBGAURARawP$, Cliff $\delta=\HeadRefitNoEBGAURACliff$).
This comparison matches the mean unperturbed detection rate but changes
both gains and the classification head; it does not isolate the EBG effect.

\begin{table}[h]
  \caption{Untuned 39-parameter classical comparison, mean $\pm$ sample SD over ten
  seeds.  AURA uses all 30 positive held-out trajectories per seed.}
  \label{tab:classical-control}
  \centering
  \small
  \begin{tabular}{lccc}
    \toprule
    Model & Parameters & Clean accuracy & Positive AURA\\
    \midrule
    \ClassicalControlRows
    \bottomrule
  \end{tabular}
\end{table}

\section{Checks of data and implementation}
\label{app:invariants}

Automated tests and the independent bundle audit check the following for every
applicable seed: train/test trajectory identifiers are disjoint; reducer means,
components, and scales are train-only; the reference-frame contrast is exactly
zero; generator-injected differences equal $A c_t^{(0)}$; protected meter changes are
exactly zero; the mapped projector is symmetric, idempotent, and rank two;
analytic physical-space gradients agree with finite differences; attack
coefficients remain in budget; condition switches and parameter counts match
their definitions; confidence methods and validity flags match conditions;
AURA is recomputed from curves; and all statistical summaries are reconstructed
from the one-seed JSON files.

The immutable production archive has SHA-256
\par\smallskip\noindent
\texttt{5c3ad85161dd1370c9bb697bb3d89c9f1497420a5f310cfc60272f01bb5ff3e1}.
\par\smallskip
Production used revision \texttt{\ProductionRevision}; the reporting-only tied
rank correction used \texttt{\ReportingRevision}.  A fresh extraction passes
\BundleChecks{}/\BundleChecks{} independent checks.  The code/evidence archive
is available through the retained Figshare access link:
\url{https://figshare.com/s/d572a5666b8edf739be9}.  The archive contains the
locked implementation, generated evidence summaries, audit procedures, and
paper-source reproduction commands for the reviewed version.  This private
access link is not a permanent publication identifier.

The positive-only execution used revision
\texttt{\PositiveRevision}; reporting used
\texttt{\PositiveReportingRevision}; and its sealed archive has SHA-256
\texttt{\PositiveArchiveSHA}.  Ingestion verifies ten worker manifests, 20
worker source hashes, clean tracked state before and after execution, exact
package versions, all merged-file hashes, the restricted source delta, and an
independent reconstruction of the subset, label-preservation checks, and 12 contrasts.

The matched-Lipschitz control used revision
\texttt{\MatchedControlRevision}; its sealed archive has SHA-256
\par\smallskip\noindent
\texttt{\MatchedArchiveSHA}.
\par\smallskip
Local ingestion verifies the archive and merged
hashes, all ten workers, 30 worker-result hashes, a common seven-file source
inventory, clean tracked state, complete positive subsets, model-field
reproduction, budget grids, 40 masking rows, 6,000 checks that the combined attack is nonzero, and three
independently reconstructed contrasts.  The classical control used revision
\texttt{\ClassicalControlRevision}; its ten rows and source artifacts are
separately checksum-bound and independently audited.

The gain-replacement control used revision
\texttt{\GainTransplantRevision}; its sealed archive has SHA-256
\par\smallskip\noindent
\texttt{\GainTransplantArchiveSHA}.
\par\smallskip
Ingestion verifies the archive and merged hashes, ten clean-tree workers,
30 worker-result hashes, one common eight-file source inventory, the exact
decomposition, \GainTransplantMaskingPasses{} masking rows, and
\GainTransplantSemanticPasses{} checks that the combined attack is nonzero.  The ten array tasks completed
in 25:59--27:21 with recorded peak resident memory below 122~MiB.

The head-refit control used revision \texttt{\HeadRefitRevision}; its sealed
archive has SHA-256
\par\smallskip\noindent
\texttt{\HeadRefitArchiveSHA}.
\par\smallskip
All ten tasks and the dependent merge completed with exit 0.  Task wall times
were 25:37--26:44 and batch peak resident memory was 116,796--126,256~KiB.
Local ingestion verifies the archive and merged hashes, ten clean-tree workers,
20 worker-result hashes, a common nine-file source inventory, unchanged parameters outside the head, changed head hashes, \HeadRefitMaskingPasses{} masking rows, and
\HeadRefitSemanticPasses{} checks that the combined attack is nonzero.

The implementation uses Qiskit 2.5.2 under Apache License 2.0.  MATPOWER
software is distributed under the 3-clause BSD license, while its case files
are expressly outside that software license.  The branch parameters used here
come from the public IEEE 14-bus case converted by MATPOWER from the University
of Washington Power Systems Test Case Archive; we cite the MATPOWER provenance
paper and redistribute only the derived Jacobian construction needed for this
generated benchmark.

\section{Reproduction commands}
\label{app:reproduce}

From the \texttt{seqformer\_pkg} source-tree root, the following commands
recreate the paper assets and perform the focused validation.  The production ablation command is
shown for completeness; its output directory should be changed for a fresh
run.

\begin{verbatim}
python3 -m venv .repro-venv
.repro-venv/bin/python -m pip install -r ../requirements-lock.txt

PYTHONPATH=. .repro-venv/bin/python \
  experiments/audit_temporal_fdia_dataset.py \
  --seeds 10 --n 1000 --qubits 4 --tokens 6 \
  --attack-min 0.10 --attack-max 0.25 --out NEW_GATE.json

PYTHONPATH=. .repro-venv/bin/python \
  experiments/run_ablation.py --scale dev \
  --dataset temporal --evaluation-class all \
  --conditions full,no_ebg,no_fgqa,no_nach,no_all \
  --seeds 10 --seed-start 3000 --n-samples 200 --steps 20 \
  --eval-n 20 --attack-steps 10 --eps-points 5 \
  --attacks pgd,spsa,eot --cert-eval-n 10 --cert-samples 100 \
  --cert-radius-points 6 --masking-samples 1 --out-dir NEW_RUN

PYTHONPATH=. .repro-venv/bin/python \
  experiments/run_ablation.py --scale dev \
  --dataset temporal --evaluation-class positive \
  --conditions full,no_ebg,no_fgqa,no_nach,no_all \
  --seeds 10 --seed-start 3000 --n-samples 200 --steps 20 \
  --eval-n 20 --attack-steps 10 --eps-points 5 \
  --attacks pgd,spsa,eot --cert-eval-n 10 --cert-samples 100 \
  --cert-radius-points 6 --masking-samples 1 --out-dir NEW_POSITIVE

PYTHONPATH=. .repro-venv/bin/python \
  experiments/audit_positive_evasion_results.py \
  ../evidence/p4-dev10-a45b1c9-20260824-a/derived NEW_POSITIVE

PYTHONPATH=. .repro-venv/bin/python \
  experiments/audit_matched_lipschitz_results.py \
  ../evidence/p4-lmatch10-6dd7a14-20260827-a/merged \
  --accepted-positive-dir \
  ../evidence/p4-positive10-eeadd6a-20260825-a \
  --expected-seeds 10

PYTHONPATH=. .repro-venv/bin/python \
  experiments/audit_classical_attention_control.py \
  ../evidence/classical-attention39-1cd1e70-20260827-a

PYTHONPATH=. .repro-venv/bin/python \
  experiments/audit_gain_transplant_results.py \
  ../evidence/p4-gtrans10-0f0b0bd-20260828-a/merged \
  --accepted-matched-dir \
  ../evidence/p4-lmatch10-6dd7a14-20260827-a \
  --expected-seeds 10

PYTHONPATH=. .repro-venv/bin/python \
  experiments/audit_head_refit_transplant_results.py \
  ../evidence/p4-hrefit10-339da07-20260829-a/merged \
  --expected-revision \
  339da0720f42a52f62fa3fc505e10b2791e7eafb \
  --expected-seeds 10

PYTHONPATH=. .repro-venv/bin/python -m pytest -q
.repro-venv/bin/python paper_satquml/generate_assets.py
.repro-venv/bin/python paper_satquml/generate_positive_assets.py \
  --evidence-dir ../evidence/p4-positive10-eeadd6a-20260825-a
PYTHONPATH=. .repro-venv/bin/python \
  paper_satquml/generate_sensitivity_control_assets.py \
  --matched-evidence-dir \
  ../evidence/p4-lmatch10-6dd7a14-20260827-a \
  --classical-evidence-dir \
  ../evidence/classical-attention39-1cd1e70-20260827-a
PYTHONPATH=. .repro-venv/bin/python \
  paper_satquml/generate_gain_transplant_assets.py \
  --gain-evidence-dir \
  ../evidence/p4-gtrans10-0f0b0bd-20260828-a \
  --matched-evidence-dir \
  ../evidence/p4-lmatch10-6dd7a14-20260827-a
PYTHONPATH=. .repro-venv/bin/python \
  paper_satquml/generate_head_refit_assets.py \
  --evidence-dir ../evidence/p4-hrefit10-339da07-20260829-a
.repro-venv/bin/python paper_satquml/build_paper.py
\end{verbatim}

\end{document}

%% file: generated/results_macros.tex
\newcommand{\ProductionRevision}{a45b1c9dea401215d6857517c7aef181a4972929}
\newcommand{\ReportingRevision}{30d7419289533e3fee7efa6f30057b25d2696253}

\newcommand{\NumRows}{50}

\newcommand{\EvaluationReuseRows}{10}
\newcommand{\DistinctAttackEvaluationRows}{40}
\newcommand{\BundleChecks}{1,834}
\newcommand{\TemporalGateMean}{0.792}
\newcommand{\TemporalGateSD}{0.022}
\newcommand{\FlatLogisticGateMean}{0.495}
\newcommand{\FlatLogisticGateSD}{0.031}
\newcommand{\FlatMLPGateMean}{0.770}
\newcommand{\FlatMLPGateSD}{0.042}
\newcommand{\FlatForestGateMean}{0.785}
\newcommand{\FlatForestGateSD}{0.014}
\newcommand{\FullClean}{0.718}
\newcommand{\FullCleanSD}{0.047}
\newcommand{\FullAURA}{0.358}
\newcommand{\FullAURASD}{0.089}

\newcommand{\NoNACHRadius}{0.15956}
\newcommand{\NoNACHRadiusSD}{0.00349}

\newcommand{\FullGammaSum}{1.567}
\newcommand{\FullGammaSumSD}{0.034}
\newcommand{\FullLipschitz}{3.135}
\newcommand{\FullLipschitzSD}{0.067}

\newcommand{\PGDWinners}{3{,}198}
\newcommand{\SPSAWinners}{2}
\newcommand{\EOTWinners}{0}
\newcommand{\TotalWinners}{3{,}200}
\newcommand{\FullMonotone}{9}
\newcommand{\NoEBGMonotone}{9}
\newcommand{\NoFGQAMonotone}{8}
\newcommand{\NoNACHMonotone}{9}
\newcommand{\NoAllMonotone}{9}
\newcommand{\IndependentMonotoneTotal}{35}
\newcommand{\IndependentNonMonotoneTotal}{5}
\newcommand{\EBGCleanDelta}{-0.02167}
\newcommand{\EBGCleanCILow}{-0.04833}
\newcommand{\EBGCleanCIHigh}{0.01500}
\newcommand{\EBGCleanHolmP}{1}

\newcommand{\EBGAURACILow}{-0.07500}
\newcommand{\EBGAURACIHigh}{-0.03000}
\newcommand{\EBGAURAHolmP}{0.035156}

\newcommand{\FGQARadiusHolmP}{0.27344}

\newcommand{\EBGAURAReduction}{0.04938}

\newcommand{\FGQARadiusDeltaSci}{-2.67\times10^{-5}}

\newcommand{\FullAtMid}{0.300}
\newcommand{\FullAtMax}{0.065}
\newcommand{\NoEBGAtMid}{0.225}
\newcommand{\NoEBGAtMax}{0.050}

%% file: generated/endpoint_rows.tex
\newcommand{\EndpointRows}{%
Full & 0.718$\pm$0.047 & 0.358$\pm$0.089 & 0.12947$\pm$0.00283 & CP bounded \\
$-$EBG & 0.697$\pm$0.071 & 0.308$\pm$0.088 & 0.05071$\pm$0.00000 & CP bounded \\
$-$FGQA & 0.718$\pm$0.047 & 0.358$\pm$0.089 & 0.12944$\pm$0.00281 & CP bounded \\
$-$NACH & 0.718$\pm$0.047 & 0.358$\pm$0.089 & 0.15956$\pm$0.00349 & invalid point$^\dagger$ \\
No modules & 0.697$\pm$0.071 & 0.308$\pm$0.088 & 0.06250$\pm$0.00000 & invalid point$^\dagger$ \\
}

%% file: generated/contrast_rows.tex
\newcommand{\ContrastRows}{%
$-$EBG & Clean & -0.02167 & $[-0.04833, 0.01500]$ & 1 & -0.27 \\
 & AURA$^\ast$ & -0.04938 & $[-0.07500, -0.03000]$ & 0.035156 & -0.32 \\
 & Mean $R$ & -0.07875 & $[-0.08095, -0.07741]$ & 0.023438 & -1.00 \\
$-$FGQA & Clean & 0.00000 & $[0.00000, 0.00000]$ & 1 & 0.00 \\
 & AURA$^\ast$ & 0.00000 & $[0.00000, 0.00000]$ & 1 & 0.00 \\
 & Mean $R$ & -0.00003 & $[-0.00005, -0.00001]$ & 0.27344 & -0.05 \\
$-$NACH & Clean & 0.00000 & $[0.00000, 0.00000]$ & 1 & 0.00 \\
 & AURA$^\ast$ & 0.00000 & $[0.00000, 0.00000]$ & 1 & 0.00 \\
 & Mean $R^\dagger$ & 0.03009 & $[0.02978, 0.03060]$ & 0.023438 & 1.00 \\
No modules & Clean & -0.02167 & $[-0.04833, 0.01500]$ & 1 & -0.27 \\
 & AURA$^\ast$ & -0.04938 & $[-0.07500, -0.03000]$ & 0.035156 & -0.32 \\
 & Mean $R^\dagger$ & -0.06697 & $[-0.06916, -0.06562]$ & 0.023438 & -1.00 \\
}

%% file: generated/curve_rows.tex
\newcommand{\RobustCurveRows}{%
Full & $.725\pm.109$ & $.590\pm.145$ & $.300\pm.125$ & $.145\pm.069$ & $.065\pm.053$ \\
$-$EBG & $.715\pm.097$ & $.540\pm.126$ & $.225\pm.132$ & $.085\pm.067$ & $.050\pm.062$ \\
$-$FGQA & $.725\pm.109$ & $.590\pm.145$ & $.300\pm.125$ & $.145\pm.069$ & $.065\pm.053$ \\
$-$NACH & $.725\pm.109$ & $.590\pm.145$ & $.300\pm.125$ & $.145\pm.069$ & $.065\pm.053$ \\
No modules & $.715\pm.097$ & $.540\pm.126$ & $.225\pm.132$ & $.085\pm.067$ & $.050\pm.062$ \\
}

%% file: generated/positive_results_macros.tex
\newcommand{\PositiveRevision}{eeadd6af64d68745240134a704c5e9701fc144c4}
\newcommand{\PositiveReportingRevision}{3c091bec2c3d6ed2d125eea6be2024e006c33da0}
\newcommand{\PositiveArchiveSHA}{2af84a262b6ce1589471e40339be8ee68d8f134f9809774bbd367741040ed763}

\newcommand{\PositiveEvaluationReuseRows}{10}
\newcommand{\PositiveDistinctAttackEvaluationRows}{40}
\newcommand{\PositiveSemanticPasses}{5{,}000/5{,}000}
\newcommand{\PositiveDistinctSemanticPasses}{4{,}000/4{,}000}
\newcommand{\PositiveSemanticReuseChecks}{1{,}000}
\newcommand{\PositiveMinSignal}{0.012225758}
\newcommand{\PositiveEBGAURALowerSeeds}{6}
\newcommand{\PositiveEBGAURAEqualSeeds}{1}
\newcommand{\PositiveEBGAURAHigherSeeds}{3}

\newcommand{\PositiveEBGAURADelta}{-0.04063}

\newcommand{\PositiveEBGAURARawP}{0.12109375}
\newcommand{\PositiveEBGAURAHolmP}{0.96875}
\newcommand{\PositiveEBGAURACliff}{-0.26}

\newcommand{\PositiveFullAtMid}{0.410}

\newcommand{\PositiveFullAtMax}{0.175}

\newcommand{\PositiveNoEBGAtMid}{0.370}

\newcommand{\PositiveNoEBGAtMax}{0.120}

%% file: generated/positive_endpoint_rows.tex
\newcommand{\PositiveEndpointRows}{%
Full & 0.419$\pm$0.079 & 0.655$\pm$0.064 & 0.175$\pm$0.054 \\
$-$EBG & 0.378$\pm$0.082 & 0.685$\pm$0.106 & 0.120$\pm$0.079 \\
$-$FGQA & 0.419$\pm$0.079 & 0.655$\pm$0.064 & 0.175$\pm$0.054 \\
$-$NACH & 0.419$\pm$0.079 & 0.655$\pm$0.064 & 0.175$\pm$0.054 \\
No modules & 0.378$\pm$0.082 & 0.685$\pm$0.106 & 0.120$\pm$0.079 \\
}

%% file: generated/positive_contrast_rows.tex
\newcommand{\PositiveContrastRows}{%
$-$EBG & Clean & -0.02167 & $[-0.04833, 0.01500]$ & 1 & -0.27 \\
 & Positive AURA & -0.04063 & $[-0.07750, -0.00625]$ & 0.96875 & -0.26 \\
 & Mean $R$ & -0.07875 & $[-0.08095, -0.07741]$ & 0.023438 & -1.00 \\
$-$FGQA & Clean & 0.00000 & $[0.00000, 0.00000]$ & 1 & 0.00 \\
 & Positive AURA & 0.00000 & $[0.00000, 0.00000]$ & 1 & 0.00 \\
 & Mean $R$ & -0.00003 & $[-0.00005, -0.00001]$ & 0.35156 & -0.05 \\
$-$NACH & Clean & 0.00000 & $[0.00000, 0.00000]$ & 1 & 0.00 \\
 & Positive AURA & 0.00000 & $[0.00000, 0.00000]$ & 1 & 0.00 \\
 & Mean $R^\dagger$ & 0.03009 & $[0.02978, 0.03060]$ & 0.023438 & 1.00 \\
No modules & Clean & -0.02167 & $[-0.04833, 0.01500]$ & 1 & -0.27 \\
 & Positive AURA & -0.04063 & $[-0.07750, -0.00625]$ & 0.96875 & -0.26 \\
 & Mean $R^\dagger$ & -0.06697 & $[-0.06916, -0.06562]$ & 0.023438 & -1.00 \\
}

%% file: generated/positive_curve_rows.tex
\newcommand{\PositiveCurveRows}{%
Full & $.655\pm.064$ & $.555\pm.107$ & $.410\pm.099$ & $.295\pm.096$ & $.175\pm.054$ \\
$-$EBG & $.685\pm.106$ & $.520\pm.116$ & $.370\pm.123$ & $.220\pm.103$ & $.120\pm.079$ \\
$-$FGQA & $.655\pm.064$ & $.555\pm.107$ & $.410\pm.099$ & $.295\pm.096$ & $.175\pm.054$ \\
$-$NACH & $.655\pm.064$ & $.555\pm.107$ & $.410\pm.099$ & $.295\pm.096$ & $.175\pm.054$ \\
No modules & $.685\pm.106$ & $.520\pm.116$ & $.370\pm.123$ & $.220\pm.103$ & $.120\pm.079$ \\
}

%% file: generated/sensitivity_control_macros.tex
\newcommand{\MatchedControlRevision}{6dd7a14cdc3808bfc19a97c5e96afc30e558b093}
\newcommand{\MatchedArchiveSHA}{632d892d4940470076abca49138f6d21cbc0c131324955ebb8a102ce598153b2}
\newcommand{\ClassicalControlRevision}{1cd1e70db38aa05cf001dca80ed5049470216dc8}
\newcommand{\MatchedControlDecision}{The AURA ordering changes under Lipschitz matching, demonstrating normalization dependence; neither tested axis alone identifies an EBG effect.}
\newcommand{\MatchedGlobalScale}{0.39187472}
\newcommand{\MatchedGlobalDelta}{0.14375}
\newcommand{\MatchedGlobalCILow}{0.11417}
\newcommand{\MatchedGlobalCIHigh}{0.16917}
\newcommand{\MatchedGlobalRawP}{0.001953125}
\newcommand{\MatchedGlobalCliff}{0.92}
\newcommand{\MatchedPairedDelta}{0.14375}
\newcommand{\MatchedPairedCILow}{0.11458}
\newcommand{\MatchedPairedCIHigh}{0.16958}
\newcommand{\MatchedPairedRawP}{0.001953125}
\newcommand{\MatchedUnmatchedDelta}{-0.03667}
\newcommand{\MatchedUnmatchedRawP}{0.099609375}
\newcommand{\ControlQuantumClean}{0.718}
\newcommand{\ControlQuantumCleanSD}{0.047}
\newcommand{\ControlQuantumAURA}{0.395}
\newcommand{\ControlQuantumAURASD}{0.052}
\newcommand{\ControlClassicalClean}{0.527}
\newcommand{\ControlClassicalCleanSD}{0.041}
\newcommand{\ControlClassicalAURA}{0.338}
\newcommand{\ControlClassicalAURASD}{0.061}
\newcommand{\ClassicalCleanDelta}{-0.19167}

\newcommand{\ClassicalCleanRawP}{0.001953125}

\newcommand{\ClassicalAURADelta}{-0.05708}

\newcommand{\ClassicalAURARawP}{0.1796875}

%% file: generated/matched_control_rows.tex
\newcommand{\MatchedControlRows}{%
Full, nominal grid & 0.395$\pm$0.052 & 0.150$\pm$0.045 & 0.1500 \\
$-$EBG, unmatched grid & 0.359$\pm$0.046 & 0.107$\pm$0.062 & 0.1500 \\
$-$EBG, global $L_{\rm enc}$ match & 0.539$\pm$0.070 & 0.417$\pm$0.085 & 0.0588 \\
$-$EBG, paired $L_{\rm enc}$ match & 0.539$\pm$0.070 & 0.417$\pm$0.085 & 0.0588 \\
}

%% file: generated/classical_control_rows.tex
\newcommand{\ClassicalControlRows}{%
Quantum attention & 39 & \ControlQuantumClean$\pm$\ControlQuantumCleanSD & \ControlQuantumAURA$\pm$\ControlQuantumAURASD \\
Classical attention & 39 & \ControlClassicalClean$\pm$\ControlClassicalCleanSD & \ControlClassicalAURA$\pm$\ControlClassicalAURASD \\
}

%% file: generated/gain_transplant_macros.tex
\newcommand{\GainTransplantRevision}{0f0b0bd61f96d305be9132a502638d8f4534fa32}
\newcommand{\GainTransplantArchiveSHA}{cd03a76626839abb830fec1826c45d9504b8159433ca4e04b73cbdb9cd8bfdfe}
\newcommand{\GainTransplantAURA}{0.012}
\newcommand{\GainTransplantAURASD}{0.021}

\newcommand{\GainTransplantMaskingPasses}{10/10}
\newcommand{\GainTransplantSemanticPasses}{1,500/1,500}

\newcommand{\GainComponentDelta}{0.34708}

\newcommand{\GainComponentHolmP}{0.00390625}
\newcommand{\GainResidualDelta}{-0.38375}

%% file: generated/gain_transplant_rows.tex
\newcommand{\GainTransplantRows}{%
Total: $N-F$ (accepted) & -0.03667 & [-0.06833, -0.00625] & 0.099609 & -- & -0.36 \\
Gain replacement: $N-T$ & 0.34708 & [0.32292, 0.37583] & 0.0019531 & 0.0039062 & 1.00 \\
Residual: $T-F$ & -0.38375 & [-0.41958, -0.34792] & 0.0019531 & 0.0039062 & -1.00 \\
}

%% file: generated/head_refit_macros.tex
\newcommand{\HeadRefitRevision}{339da0720f42a52f62fa3fc505e10b2791e7eafb}
\newcommand{\HeadRefitArchiveSHA}{1bc8cac68e5287cb7c21cfc781985e6ab54236ff0ff8451eedf41a2236440856}
\newcommand{\HeadRefitAURA}{0.388}
\newcommand{\HeadRefitAURASD}{0.046}
\newcommand{\HeadRefitPositiveClean}{0.663}
\newcommand{\HeadRefitPositiveCleanSD}{0.078}
\newcommand{\HeadRefitTestClean}{0.733}
\newcommand{\HeadRefitTestCleanSD}{0.062}
\newcommand{\HeadRefitAgreement}{0.890}

\newcommand{\HeadRefitMaskingPasses}{10/10}
\newcommand{\HeadRefitSemanticPasses}{1,500/1,500}
\newcommand{\HeadRefitFrozenAURADelta}{0.37625}
\newcommand{\HeadRefitFrozenAURACILow}{0.34542}
\newcommand{\HeadRefitFrozenAURACIHigh}{0.40792}

\newcommand{\HeadRefitNoEBGAURADelta}{0.02917}
\newcommand{\HeadRefitNoEBGAURACILow}{0.00667}
\newcommand{\HeadRefitNoEBGAURACIHigh}{0.05583}
\newcommand{\HeadRefitNoEBGAURARawP}{0.08203125}
\newcommand{\HeadRefitNoEBGAURACliff}{0.32}

\newcommand{\HeadRefitNoEBGCleanDelta}{-0.00000}
\newcommand{\HeadRefitNoEBGCleanCILow}{-0.05333}
\newcommand{\HeadRefitNoEBGCleanCIHigh}{0.04333}